\documentclass{mystyle}
\title[Maximum overlap of polyhedron and polygon]{Maximum overlap area of a convex polyhedron and a convex polygon under translation}
\author[H.\,J. Kweon]{Hyuk Jun Kweon}
\author[H. Zhu]{Honglin Zhu}

\begin{document}
\begin{abstract}
    Let $P$ be a convex polyhedron and $Q$ be a convex polygon with $n$ vertices in total in three-dimensional space. We present a deterministic algorithm that finds a translation vector $v \in \mathbb{R}^3$ maximizing the overlap area $|P \cap (Q + v)|$ in $O(n \log^2 n)$ time. We then apply our algorithm to solve two related problems. We give an $O(n \log^3 n)$ time algorithm that finds the maximum overlap area of three convex polygons with $n$ vertices in total. We also give an $O(n \log^2 n)$ time algorithm that minimizes the symmetric difference of two convex polygons under scaling and translation.
\end{abstract}

\maketitle

\maketitle

\setcounter{page}{1}

\setcounter{tocdepth}{1}
\tableofcontents

\section{Introduction}\label{sec_intro}
Shape matching is an important topic in computational geometry, with useful applications in areas such as computer graphics. In a typical problem of shape matching, we are supplied two or more shapes, and we want to determine how much the shapes resemble each other. More precisely, given a similarity measure and a set of allowed transformations, we want to transform the shapes to maximize their similarity measure. 

There are many candidates for the similarity measure, such as the Hausdorff distance and the Fr\'echet distance between the boundaries of the shapes. We can also consider the area/volume of overlap or of symmetric difference. The advantage to these is that they are more robust against noise on the boundary of the shapes \cite{deberg1996}.

The maximum overlap problem of convex polytopes has been studied by many. In dimension $2$, de Berg et al. \cite{deberg1996} give an $O(n \log n)$ time algorithm for finding a translation maximizing the area of intersection of two convex polygons (where $n$ denotes the total number of vertices of the polygons). In dimension $3$, Ahn et al. \cite{ahn2008} give an $O(n^3 \log^4 n)$ expected time algorithm finding the maximum overlap of two convex polyhedra under translation. For the same problem, Ahn et al. \cite{ahn2013} present an algorithm that runs in $O(n \log^{3.5} n)$ time with probability $1 - n^{-O(1)}$ and an additive error. For $d > 3$, given two convex polytopes of dimension $d$ with $n$ facets in total, Ahn et al. \cite{ahn2013} give an algorithm that finds the maximum overlap under translation in $O(n^{\lfloor d/2 \rfloor + 1} \log^{d} n)$ time with probability $1 - n^{O(1)}$ and an additive error.

In the plane, when all rigid motions are allowed, Ahn et al. \cite{ahn2007} give an approximate algorithm that finds a rigid motion realizing at least $1-\epsilon$ times the maximal overlap in $O((1/\epsilon)\log n + (1/\epsilon^2) \log (1/\epsilon))$ time. In dimension $3$, Ahn et al. \cite{ahn2014} present an approximate algorithm that finds a rigid motion realizing at least $1-\epsilon$ times the maximal overlap in $O(\epsilon^{-3} n \log^{3.5} n)$ with probability $1 - n^{-O(1)}$.

When considering the maximum overlap as a similarity measure, we obviously can only allow area/volume-preserving transformations. However, we may want to allow scaling as a transformation---two similar triangles are supposed to be very ``similar,'' though they may have different sizes. In this case, the area of symmetric difference is a better measure of similarity. Yon et al. \cite{yon2016} give an algorithm minimizing the symmetric difference of two convex polygons under translation and scaling in $O(n \log^3 n)$ expected time.

\subsection*{Our results}
While many have studied the matching problem for two convex polytopes of the same dimension, to our knowledge no one has examined the problem for polytopes of different dimensions or matching more than two polytopes. 

The main result in this paper is a deterministic algorithm for the problem of matching a convex polyhedron and a convex polygon under translation in three-dimensional space. 

\begin{restatable}{theorem}{algo}\label{thm_algo}
    Let $P$ be a convex polyhedron and $Q$ a convex polygon with $n$ vertices in total. We can find a vector $v \in \mathbb{R}^3$ that maximizes the overlap area $|P \cap (Q+v)|$ in $O(n \log^2 n)$ time.
\end{restatable}

We also present two applications of our algorithm to other problems in computational geometry. First, we give a deterministic algorithm for maximizing the overlap of three convex polygons under translations.

\begin{restatable}{theorem}{threepolys}\label{thm_three_polygons}
    Let $P$, $Q$, $R$ be three convex polygons with $n$ vertices in total in the plane. We can find a pair of translations $(v_Q, v_R) \in \mathbb{R}^4$ that maximizes the overlap area $|P \cap (Q + v_Q) \cap (R + v_R)|$ in $O(n\log^3 n)$ time. 
\end{restatable}

We also give a deterministic $O(n \log^2 n)$ time algorithm for minimizing the symmetric difference of two convex polygons under a homothety (a translation and a scaling), which is an improvement to Yon et al.'s randomized algorithm \cite{yon2016}.

\begin{restatable}{theorem}{symmdiff}\label{thm_symmetric_difference}
    Let $P$ and $Q$ be convex polygons with $n$ vertices in total. Then we can find a homothety $\varphi$ that minimizes the area of symmetric difference $|P \setminus \varphi(Q)| + |\varphi(Q) \setminus P|$ in $O(n \log^2 n)$ time. 
\end{restatable}

The main ingredient in the proof of \Cref{thm_algo} is a new technique we introduce which generalizes Megiddo's prune-and-search \cite{megiddo1984}. This  allows us to efficiently prune among $n$ groups of $m$ parallel lines. 

\begin{restatable}{theorem}{pruneandsearch}\label{thm_prune_and_search}
    Let $S = \bigcup_{i=1}^{n} S_i$ be a union of $n$ sets of $O(m)$ parallel lines in the plane, none of which are parallel to the $x$-axis, and suppose the lines in each $S_i$ are indexed from left to right.
    
    Suppose there is an unknown point $p^{*} \in \mathbb{R}^2$ and we are given an oracle that decides in time $T$ the relative position of $p^{*}$ to any line in the plane. Then we can find the relative position of $p^{*}$ to every line in $S$ in $O(n \log^2 m + (T + n) \log(mn))$ time. 
\end{restatable}

\subsection*{Organization of the Paper}
In \Cref{sec_prelim}, we introduce the problem of matching a convex polyhedron and a convex polygon under translation in three-dimensional space. In \Cref{sec_technique}, we present a core technique we use in our algorithm, which is a generalization of Megiddo's prune-and-search technique \cite{megiddo1984}. In \Cref{sec_three_and_two}, we present the algorithm for \Cref{thm_algo}. In \Cref{sec_three_polygons}, we apply our algorithm to solve the problem of maximizing the intersection of three polygons under translation. In \Cref{sec_symmetric_difference}, we give the algorithm for minimizing the symmetric difference of two convex polygons under homothety. 

\subsection*{Acknowledgements}
This paper is the result of the MIT SPUR 2022, a summer undergraduate research program organized by the MIT math department. The authors would like to thank the faculty advisors David Jerison and Ankur Moitra for their support and the math department for providing this research opportunity. We thank the anonymous referees of SoCG 2023 for providing helpful comments that increased the quality of this paper.

\section{Preliminaries}\label{sec_prelim}
Let $P \subset \mathbb {R}^3$ be a convex polyhedron and $Q \subset \mathbb{R}^2$ be a convex polygon with $n$ vertices in total. Throughout the paper, we assume that $Q$ is in the $xy$-plane, and that the point in $P$ with minimal $z$ coordinate is on the $xy$-plane. We want to find a translation vector $v = (x, y, z) \in \mathbb{R}^3$ that maximizes the overlap area $f(v) = |P \cap (Q + v)|$. 

It is easy to observe that $f(v)$ is continuous and piecewise quadratic on the interior of its support. As noted in \cite{deberg1996, ahn2008, ahn2013}, $f$ is smooth on a region $R$ if $P \cap (Q+v)$ is combinatorially equivalent for all $v \in R$, that is, if we have the same set of face-edge incidences between $P$ and $Q$. Following the convention of \cite{ahn2008}, we call the polygons that form the boundaries of these regions the \textit{event polygons}, and as in \cite{deberg1996}, we call the space of translations of $Q$ the \textit{configuration space}. The arrangement of the event polygons partition the configuration space into cells with disjoint interiors. The overlap function $f(v)$ is quadratic on each cell. Thus, to locate a translation maximizing $f$, we need to characterize the event polygons. 

For two sets $A, B \subset \mathbb{R}^d$, we write the \textit{Minkowski sum} of $A$ and $B$ as
\[
    A + B := \{a+b|a \in A, b\in B\}.
\]
We will make no distinction between the translation $A + v$ and the Minkowski sum $A + \{v\}$ for a vector $v$. We also write $A-B$ for the Minkowski sum of $A$ with $-B = \{-b | b \in B\}$. We categorize the event polygons into three types and describe them in terms of Minkowski sums:
\begin{enumerate}[label = (\Roman*)]
    \item \label{type1} When $Q + v$ contains a vertex of $P$. For each vertex $u$ of $P$, we have an event polygon $u - Q$. There are $O(n)$ event polygons of this type.
    \item \label{type2} When a vertex of $Q + v$ is contained in a face of $P$. For each face $F$ of $P$ and each vertex $v$ of $Q$, we have an event polygon $F - v$. There are $O(n^2)$ event polygons of this type.
    \item \label{type3} When an edge of $Q + v$ intersects an edge of $P$. For each edge $e$ of $P$ and each edge $e'$ of $Q$, we have an event polygon $e - e'$. There are $O(n^2)$ event polygons of this type.
\end{enumerate}

The reason that convexity is fundamental is due to the following standard fact, as noted and proved in \cite{deberg1996, yon2016}.
\begin{proposition}\label{prop_concavity}
    Let $P$ be a $d'$-dimensional convex polytope and let $Q$ be a $d$-dimensional convex polytope. Suppose $d' \geq d$. Let $f(v) = \operatorname{Vol}(P \cap (Q + v))$ be the volume of the overlap function. Then, $f(v)^{1/d}$ is concave on its support $\operatorname{supp}(f) = \{v|f(v) > 0\}$. 
\end{proposition}
As in \cite{avis1996}, we say a function $f: \mathbb{R} \to \mathbb{R}$ is \textit{unimodal} if it increases to a maximum value, possibly stays there for some interval, and then decreases. It is \textit{strictly unimodal} if it strictly increases to the maximum and then strictly decreases. Furthermore, we say a function $f: \mathbb{R}^d \to \mathbb{R}$ is (strictly) unimodal if its restriction to any line is (strictly) unimodal.

The following corollary of \Cref{prop_concavity} allows us to employ a divide-and-conquer strategy in our algorithm. 
\begin{corollary}[\cite{avis1996}]\label{cor_unimodality}
For any line $l$ parameterized by $l = p + vt$ in $\mathbb{R}^{d'}$ for $v \neq 0$, the function $f_l(t) = f(p + vt)$ is strictly unimodal.
\end{corollary}

We also use the following two techniques in our algorithm. 
\begin{lemma}[\cite{frederickson1984}]\label{lemma_sorted_matrix}
    Let $M$ be an $m \times n$ matrix of real numbers, where $m \leq n$. If every row and every column of $M$ is in increasing order, then we say $M$ is a \textit{sorted matrix}. For any positive integer $k$ smaller or equal to $mn$, the $k$-th smallest entry of $M$ can be found in $O(m \log(2n/m))$ time, assuming an entry of $M$ can be accessed in $O(1)$ time. 
\end{lemma}
For our purposes, we will use this result in the weaker form of $O(m+n)$.

\begin{lemma}[\cite{chazelle1993}]\label{lemma_cutting}
    Given $n$ hyperplanes in $\mathbb{R}^d$ and a region $R \subset \mathbb{R}^d$, a \textit{$(1/r)$-cutting} is a collection of simplices with disjoint interiors, which together cover $R$ and such that the interior of each simplex intersects at most $n/r$ hyperplanes. A $(1/r)$-cutting of size $O(r^d)$ can be computed deterministically in $O(nr^{d-1})$ time. In addition, the set of hyperplanes intersecting each simplex of the cutting is reported in the same time. 
\end{lemma}

\section{Generalized two-dimensional prune-and-search}\label{sec_technique}
In this section, we prove \Cref{thm_prune_and_search}, our generalization of Megiddo's prune-and-search technique \cite{megiddo1984}. This technique is of independent interest and can likely be applied to other problems.

In \cite{megiddo1984}, Megiddo proves the following:
\begin{theorem}[\cite{megiddo1984}]\label{thm_megiddo}
    Suppose there exists a point $p^{*} \in \mathbb{R}^2$ not known to us. Suppose further that we have an oracle that can tell us for any line $l \subset \mathbb{R}^2$ whether $p^{*} \in l$, and if $p^{*} \notin l$, the side of $l$ that $p^{*}$ belongs to. Let $T$ be the running time of the oracle. Then given $n$ lines in the plane, we can find the position of $p^{*}$ relative to each of the $n$ lines in $O(n + T \log n)$ time.
\end{theorem}

We are interested in a generalized version of Megiddo's problem. Suppose, instead of $n$ lines, we are given $n$ sets of parallel lines $S_1, S_2, \ldots, S_n$, each of size $O(m)$. In addition, suppose the lines in each $S_i$ are indexed from left to right (assuming none of the lines are parallel to the $x$-axis). Again, we want to know the position of $p^{*}$ relative to every line in $S = \bigcup_{i=1}^{n} S_i$. Megiddo's algorithm solves this problem in $O(mn + T \log(mn))$ time, but we want a faster algorithm for large $m$ by exploiting the structure of $S$.

Without loss of generality, suppose that there are no lines parallel to the $y$-axis. For each $i$ between $1$ and $n$, suppose $S_i = \{l_{i}^{j} | l_{i}^{a} \text{ lies strictly to the left of } l_{i}^{b} \text { iff } a<b\}$. Suppose that $p^{*} = (x^{*}, y^{*}) \in \mathbb{R}^2$. To report our final answer, we simply need to provide, for each $S_i$, the two consecutive indices $a$ and $a+1$ such that $p^{*}$ lies strictly between $l_{i}^{a}$ and $l_{i}^{a+1}$ or the single index $a$ such that $p^{*} \in l_{i}^{a}$. 
    
In our algorithm, we keep track of a feasible region $R$ containing $P^*$, which is either the interior of a (possibly unbounded) triangle or an open line segment if we find a line $l$ that $p^{*}$ lies on. Together with $R$, we keep track of the $2n$ indices $\operatorname{lower}(i)$ and $\operatorname{upper}(i)$ such that $S^R = \bigcup_{i=1}^{n} S_{i}^R = \{l_i^{j} | j \in (\operatorname{lower}(i), \operatorname{upper}(i)]\}$ is the set of lines intersecting $R$, which is also the set of lines we do not yet know the relative position to $p^{*}$. In the beginning, $R = \mathbb{R}^2$. Each step, we find $O(1)$ lines to run the oracle on to find a new feasible region $R' \subset R$ such that $|S^R| \leq 17/18 |S^{R'}|$ and recurse on $R'$. An outline is given in \Cref{pseudo_prune_and_search}.
    
\begin{algorithm}[ht]
    \DontPrintSemicolon
    \SetKwInOut{Input}{input}\SetKwInOut{Output}{output}
    \Input{A set $S = \bigcup_{i=1}^{n} S_{i} = \{l_{i}^{j}\}$ of $O(mn)$ lines}
    \Output{A list of indices that indicate the position of $p^{*}$ to each $S_i$}
    $R \longleftarrow \mathbb{R}^2$\;
    $S^R \longleftarrow S$\;
    \While{$|S^R| \geq 18$}{
        Find $O(1)$ lines to run the oracle on\;
        Compute the piece $R' \subset R$ containing $p^{*}$\; \tcc{We guarantee that $R'$ intersects at most $17/18$ of the lines that intersect $R$}
        Triangulate $R'$ with $O(1)$ lines to run the oracle on\;
        Update $S^R \longleftarrow S^{R'}$
        }
    Compute relative position of $p^{*}$ to the remaining lines by brute force\;
    \caption{Pseudocode for \Cref{thm_prune_and_search}}\label{pseudo_prune_and_search}
\end{algorithm}
    
One extra computational effort is updating $S^{R'}$ by computing $\operatorname{lower}(i)$ and $\operatorname{upper}(i)$. Since the feasible region is always a convex set of constant complexity, we can use binary search on $S_i^{R}$ to find the new bounds for $S_i^{R'}$ in $O(\log |S_{i}^{R}|)$ time. Thus, the total time involved in this process, assuming $|S^R|$ decreases by at least $\epsilon = 1/18$ each iteration, is
\begin{align*}
    & \sum_{i=1}^n O(\log |S_{i}|) + \sum_{i=1}^n O(\log |S_{i}^{R_1}|) + \sum_{i=1}^n O(\log |S_{i}^{R_2}|) + \cdots \\
    = & O(n \log (\frac{1}{n}|S|)) + O(n \log (\frac{1}{n}|S^{R_1}|)) + O(n \log (\frac{1}{n}|S^{R_2}|)) + \cdots \\
    = & O(n \log (m)) + O(n \log (m(1-\epsilon))) + O(n \log (m(1-\epsilon)^2)) + \cdots \\
    = & O(n \log^2 m).
\end{align*}

%We first give a related definition. For $n$ distinct  sorted real numbers $x_1, \ldots, x_n$ with positive weights $w_1, \ldots, w_n$ such that $\sum_{i = 1}^n w_i = 1$, the \textit{weighted median} is the element $x_k$ with
    %\[
        %\sum_{i = 1}^{k-1}w_i \leq 1/2 \quad \text{and} \quad \sum_{i = k+1}^{n}w_i \leq 1/2.
    %\]
%If there are two elements $x_k$ and $x_{k+1}$ satisfying the condition, then the weighted median is the mean of $x_k$ and $x_{k+1}$. In the case the weights are all $1/n$, this is just the ordinary median. 
    
We will use the following well-known result:
\begin{lemma}[\cite{cormen2009}]\label{lemma_weighted_median}
    Suppose we are given $n$ distinct real numbers with positive weights that sum to $1$. Then we can find the weighted median of these numbers in $O(n)$ time. 
\end{lemma} 

Given $S^R$ and $R$, we want to find $R' \subset R$ to recurse on.
\begin{lemma}\label{lemma_balanced_quadrants}
    If $|S^R| \geq 18$, then in $O(T + n)$ time, we can find a region $R' \subset R$ of constant complexity containing $p^*$ so that its interior intersects no more than $17/18$ of all the lines in $S^R$.
\end{lemma}
\begin{proof}
    For convenience, we write $S^R = S = \bigcup_{i=1}^{n} S_{i} = \{l_{i}^{j}\}$. We first find the weighted median of the slopes of the lines in $S$, where the slope of the lines of $S_i$ is weighted by $|S_i|/|S|$. This can be done in $O(n)$ time by \Cref{lemma_weighted_median}. 

    If this slope is equal to the slope of some line in $S_i$ and $|S_i| \geq \frac{1}{9} |S|$, then we can simply divide the plane using the median line of $S_i$ and the $x$-axis and the interior of each quadrant will avoid at least $1/18$ of the lines of $S$. 

    Otherwise, at least $4/9$ of the lines have slopes strictly greater than/less than the median slope. Without loss of generality, we assume at least $4/9$ of the lines have positive slope and at least $4/9$ of the lines have negative slope. Now let $S_{+} = \bigcup_{i=1}^{k} S_{i}$ and $S_{-} = \bigcup_{i=k+1}^{n} S_{i}$ denote the set of lines with postive/negative slope, respectively. We remove lines from the larger of the two sets until they have the same size. 

    \begin{figure}[ht]
      \centering 
      \begin{tikzpicture}[thick,scale=1.4,line join=round]
        \draw[color=topColor] (0  ,0) -- (-1.5,2);
        \draw[color=topColor] (0.5,0) -- (-1  ,2);
        \draw[color=midColor] (1  ,0) -- (-0.5,2);
        \draw[color=botColor] (1.5,0) -- ( 0.5,2);
        \draw[color=botColor] (2  ,0) -- ( 1  ,2);
        
        \draw[color=topColor] (2.5,0) -- ( 3  ,2);
        \draw[color=topColor] (3  ,0) -- ( 3.5,2);
        \draw[color=midColor] (3.5,0) -- ( 4.5,2);
        \draw[color=botColor] (4  ,0) -- ( 5  ,2);
        \draw[color=botColor] (4.5,0) -- ( 5.5,2);
        \node at (0.5 ,-0.25) {$S_1$};
        \node at (1.75,-0.25) {$S_2$};
        \node at (2.75,-0.25) {$S_3$};
        \node at (4   ,-0.25) {$S_4$};
      \end{tikzpicture}
      \caption{$P_1$, $P_2$ are $P_3$ are represented by colors.}\label{fig_line_partition}
    \end{figure}
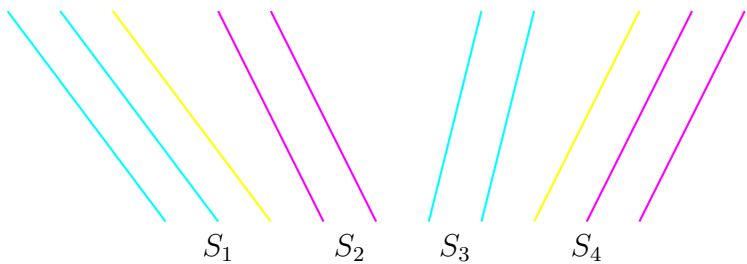

    We partition $S_{+} \cup S_{-}$ into $O(n)$ subsets $P_i$ each containing the same number of lines from $S_{+}$ and $S_{-}$ in the following way: going in lexicographical order by the indices of the lines, we put a line from $S_{1}$ and a line from $S_{k+1}$ into $P_1$ until we exhaust one of the sets (say it is $S_{k+1}$). Then, we move on to put a line from the remaining $S_{1}$ and a line from $S_{k+2}$ into $P_2$ until we exhaust one of them, and so on. Each $P_i$ is then of the form $\{l_{a(i)}^{b(i)},\ldots, l_{a(i)}^{b(i) + |P_i|/2 - 1}, l_{c(i)}^{d(i)},\ldots, l_{c(i)}^{d(i) + |P_i|/2 - 1} \}$, and can be represented by the indices $(a(i), b(i))$ and $(c(i), d(i))$ (see \Cref{fig_line_partition}). We can compute this partition in $O(n)$ time. For each $P_i$, we compute the intersection $p_i = (x_i, y_i)$ of the median line in $P_i$ with positive slope and the median line with negative slope, and assign $p_i$ a weight $w_i = |P_i| / (2|S_{+}|)$. Then, the weights of the $p_i$ sum to $1$. The significance of this is that if we know the relative position of $p^{*}$ to the lines $x = x_i$ and $y = y_i$, then we know the relative position of $p^{*}$ to at least $1/4$ of the lines in $P_i$, which is at least $\frac{2}{9} w_i$ of all the lines in $|S|$. 

    We find the median point $q = (x_{q}, y_{q})$ of the $p_i$'s by weight in $x$-coordinate in $O(n)$ time by \Cref{lemma_weighted_median}. We run the oracle on the line $x = x_{q}$. Let $p_{k_1}, p_{k_2}, \ldots, p_{k_l}$ be the points such that we now know the relative position of $p^{*}$ to $x_{k_i}$. Then the weights of these points sum to at least $1/2$. We find the median point $q' = (x_{q'}, y_{q'})$ of these by weight in $y$-coordinate in $O(n)$ time. We run the oracle on the line $y = y_{q'}$. Then, for points with weights that sum to at least $1/4$, we now know the relative position of $p^{*}$ to the vertical line and the horizontal line through those points. This means that we know the relative position of $p^{*}$ to at least $\frac{2}{9} \cdot \frac{1}{4} = \frac{1}{18}$ of all the lines in $|S|$. We get a new feasible region according to the two oracle calls whose interior avoids at least $1/18$ of the lines in $S$, and we triangulate it with $O(1)$ more oracle calls to get our desired region, in $O(T + n)$ time total. 
\end{proof}

\begin{proof}[Proof of \Cref{thm_prune_and_search}]
    After $O(\log mn)$ recursive iterations of \Cref{lemma_balanced_quadrants}, we arrive at a feasible region intersecting at most $17$ lines in $S$, and we can finish by brute force. Therefore, our algorithm runs in $O(n \log^2 m + (T + n) \log(mn))$ time. 
\end{proof}

\begin{remark}
    A simpler and probably more practical algorithm for \Cref{lemma_balanced_quadrants} is simply choosing a random line from $S_{+}$ and $S_{-}$ to intersect and run the oracle on the horizontal and vertical line through the intersection. This method gives the same run time in expectation. 
\end{remark}

\section{Maximum overlap of convex polyhedron and convex polygon}\label{sec_three_and_two}

In section, we give the algorithm that finds a translation $v \in \mathbb{R}^3$ maximizing the area of overlap function $f$. Following the convention in \cite{deberg1996}, we call such a translation a \textit{goal placement}. In the algorithm, we keep track of a closed \textit{target region} $R$ which we know contains a goal placement and decrease its size until for each event polygon $F$, either $F \cap \operatorname{interior}(R) = \varnothing$ or $F \supset R$. Then, $f$ is quadratic on $R$ and we can find the maximum of $f$ on $R$ using standard calculus. Thus, the goal of our algorithm is to efficiently trim $R$ to eliminate event polygons that intersect it. 

In the beginning of the algorithm, the target region is the interior of the Minkowski sum $P - Q$, where the overlap function is positive. By the unimodality of the overlap function, the set of goal placements is convex. Thus, for a plane in the configuration space, either it contains a goal placement, or all goal placements lie on one of the two open half spaces separated by the plane. If we have a way of knowing which case it is for any plane, we can decrease the size of our target region by cutting it with planes and finding the piece to recurse. More precisely, we need a subroutine \textbf{PlaneDecision} that decides the relative position of the set of goal placements to a plane $S$. 

Whenever \textbf{PlaneDecision} reports that a goal placement is found on a plane, we can let the algorithm terminate. Thus, we can assume it always reports a half-space containing a goal placement. 

As in \Cref{pseudo_algo}, we break down our algorithm into three stages. 

\begin{algorithm}[ht]
    \DontPrintSemicolon
    \SetKwInOut{Input}{input}\SetKwInOut{Output}{output}
    \Input{A convex polyhedron $P \in \mathbb{R}^3$ and a convex polygon $Q \in \mathbb{R}^3$ with $n$ vertices in total}
    \Output{A translation $v \in \mathbb{R}^3$ maximizing the area $|P \cap (Q + v)|$}
    Locate a horizontal slice containing a goal placement that does not contain any vertices of $P$ and replace $P$ by this slice of $P$\;
    Find a ``tube'' $D + l_y$ whose interior contains a goal placement and intersects $O(n)$ event polygons, where $D$ is a triangle in the $xz$-plane and $l_y$ is the $y$-axis\;
    Recursively construct a $(1/2)$-cutting of the target region $D + l_y$ to find a simplex containing a goal placement that does not intersect any event polygon\;
    \caption{Pseudocode for \Cref{thm_algo}}\label{pseudo_algo}
\end{algorithm}

\subsection{Stage 1}
In the first stage of our algorithm, we make use of \cite{deberg1996} to simplify our problem so that $P$ can be taken as a convex polyhedron with all of its vertices on two horizontal planes. 

We sort the vertices of $P$ by $z$-coordinate in increasing order and sort the vertices of $Q$ in counterclockwise order. Next, we trim the target region with horizontal planes (planes parallel to the $xy$-plane) to get to a slice that does not contain any vertices of $P$.
\begin{lemma}\label{lemma_stage_1}
    In $O(n \log^2 n)$ time, we can locate a strip $R = \{(x, y, z) | z \in [z_0, z_1]\}$ whose interior contains a goal placement and $P$ has no vertices with $z \in [z_0, z_1]$.
\end{lemma}
\begin{figure}[ht]
  \centering 
\tdplotsetmaincoords{100}{0}
\begin{tikzpicture}[tdplot_main_coords,thick,scale=1.4,line join=round]
  \begin{scope}
    \coordinate (A) at (-0.5,1,2);
    \coordinate (B) at (1,-1,1);
    \coordinate (C) at (-1.5,-1,0);
    \coordinate (D) at (0,2.5,1);
    \coordinate (E) at (0,-1,-1);
    
    \draw (A) -- (B) -- (E) -- (C) -- cycle;
    \draw (A) -- (D) -- (E);
    \draw (B) -- (D) -- (C);
    \draw[dotted] (B) -- (C);
  \end{scope}
  \begin{scope}[shift={(3.5,0)}]
    \coordinate (A) at (-0.5,1,2);
    \coordinate (B) at (1,-1,1);
    \coordinate (C) at (-1.5,-1,0);
    \coordinate (D) at (0,2.5,1);
    \coordinate (E) at (0,-1,-1);
    \coordinate (DE) at (0,0.75,0);
    \coordinate (BE) at (0.5,-1,0);
    \coordinate (AC) at (-1,0,1);
    \draw [draw=none, fill=topColor, fill opacity=0.5]
    (D) -- (AC) -- (B) -- cycle;
    \draw [draw=none, fill=botColor, fill opacity=0.5]
    (C) -- (DE) -- (BE) -- cycle;
    
    \draw (A) -- (B) -- (E) -- (C) -- cycle;
    \draw (A) -- (D) -- (E);
    \draw (B) -- (D) -- (C);
    \draw[dotted] (B) -- (C);
    
    \draw[dotted] (D) -- (AC) -- (B);
    \draw         (C) -- (DE) -- (BE);
    \draw[dotted] (BE) -- (C);
    
    \draw [draw=none, fill=topColor, fill opacity=0.2]
    (-2,-1.5,1) -- (1.25,-1.5,1) -- (1.75,3,1) -- (-1.5,3,1) -- cycle;
    \draw [draw=none, fill=botColor, fill opacity=0.2]
    (-2,-2,0) -- (1.25,-2,0) -- (1.75,2.5,0) -- (-1.5,2.5,0) -- cycle;
  \end{scope}
  \begin{scope}[shift={(7,0)}]
    \coordinate (A) at (-0.5,1,2);
    \coordinate (B) at (1,-1,1);
    \coordinate (C) at (-1.5,-1,0);
    \coordinate (D) at (0,2.5,1);
    \coordinate (E) at (0,-1,-1);
    
    \coordinate (AC) at (-1,0,1);
    \coordinate (DE) at (0,0.75,0);
    \coordinate (BE) at (0.5,-1,0);
    
    \draw[draw=none] (A) -- (E);
    \draw [draw=none, fill=topColor, fill opacity=0.5]
    (D) -- (AC) -- (B) -- cycle;
    \draw [draw=none, fill=botColor, fill opacity=0.5]
    (C) -- (DE) -- (BE) -- cycle;
    
    \draw (AC) -- (C) -- (DE) -- (BE) -- (B) -- cycle;
    \draw (B) -- (D) -- (C);
    \draw (AC) -- (D) -- (DE);
    \draw[dotted] (B) -- (C) -- (BE);
  \end{scope}
\end{tikzpicture}
\caption{The slice of $P$ with $z \in [z_0, z_1]$.}\label{fig_stage_1}
\end{figure}
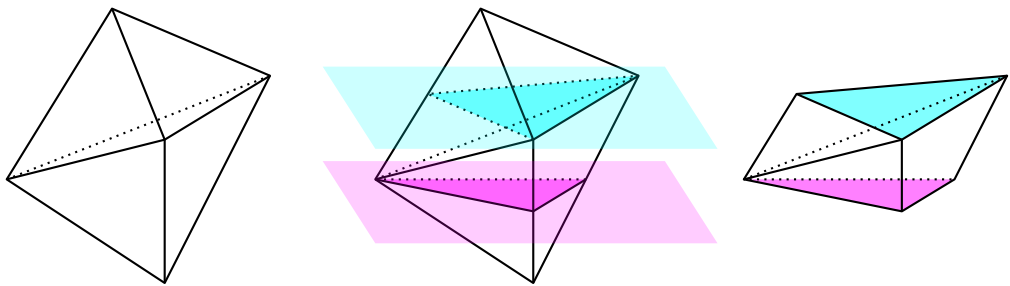
\begin{proof}
    Starting with the median $z$-coordinate of the vertices of $P$, we perform a binary search on the levels containing a vertex of $P$. For a horizontal plane $S$, \cite[Theorem 3.8]{deberg1996} allows us to compute the maximum overlap of $P\cap S$ and $Q$ under translation in $O(n\log n)$-time. The two planes $S_1$ and $S_2$ with the largest maximum values will be the bounding planes for the slice containing a goal placement by the unimodality of $f$. Thus, by a binary search, we can locate this slice in $O(n \log^2 n)$ time.
\end{proof}
By Chazelle's algorithm \cite{chazelle1992}, the convex polyhedron $P' = \{(x, y, z)\in P |z \in [z_0, z_1]\}$ can be computed in $O(n)$ time. From now on, we replace $P$ with $P'$ (see \Cref{fig_stage_1}). Without loss of generality, assume $z_0 = 0$ and $z_1 = 1$.

The region in the configuration space where $|P \cap (Q+v)| > 0$ is the Minkowski sum $P - Q$. Since $P$ only has two levels $P_0 = \{(x, y, z) \in P | z = 0\}$ and $P_1 = \{(x, y, z) \in P | z = 1\}$ that contain vertices, the Minkowski sum $P - Q$ is simply the convex hull of $(P_0 - Q) \cup (P_1 - Q)$, which has $O(n)$ vertices. We can compute $P_0 - Q$ and $P_1 - Q$ in $O(n)$ time and compute their convex hull in $O(n \log n)$ time by Chazelle's algorithm \cite{chazelle1993b}.

\subsection{PlaneDecision}
With the simplification of the problem in Stage $1$, we now show that the subroutine \textbf{PlaneDecision} can be performed in $O(n \log n)$ time. Let $S$ be a fixed plane in the configuration space. We call a translation $v$ that achieves $\operatorname{max}_{v \in S} f(v)$ a \textit{good placement}. First, we can compute the intersection of $S$ with $P - Q$ in $O(n)$ time by Chazelle's algorithm \cite{chazelle1992}. If the intersection is empty, we just report the side of $S$ containing $P - Q$. From now on assume this is not the case. 

The following lemma shows that \textbf{PlaneDecision} runs in the same time bound as the algorithm that just finds the maximum of $f$ on a plane. 
\begin{lemma}\label{lemma_side_decision}
    Suppose we can compute $\operatorname{max}_{v \in S} f(v)$ for any plane $S \subset \mathbb{R}^3$ in time $T$, then we can perform \textbf{PlaneDecision} for any plane in time $O(T)$. 
\end{lemma}
\begin{proof}
    The idea is to compute $\operatorname{max}_{v \in S'} f(v)$ for certain $S'$ that are perturbed slightly from $S$ to see in which direction relative to $S$ does $f$ increase.
    
    We compute over an extension of the reals $\mathbb{R}[\omega]/(\omega^3)$, where $\omega > 0$ is smaller than any real number. Let $A > 0$ be the maximum of $f$ over a plane $S$. Let $S_+$ and $S_-$ be the two planes parallel to $S$ that have distance $\omega$ from $S$. We compute $A_+ = \operatorname{max}_{v \in S_+} f(v)$ and $A_- = \operatorname{max}_{v \in S_-} f(v)$ in $O(T)$ time. Since $f$ is piecewise quadratic, $A_+$ and $A_-$ as symbolic expression will only involve quadratic terms in $\omega$. Since $f$ is strictly unimodal on $P - Q$, there are three possibilities:
    \begin{enumerate}
        \item If $A_+ > A$, then halfspace on the side of $S_+$ contains the set of goal placements.
        \item If $A_- > A$, then halfspace on the side of $S_-$ contains the set of goal placements.
        \item If $A \geq A_+$ and $A \geq A_-$, then $A$ is the global maximum of $f$.
    \end{enumerate}
    Thus, in $O(T)$ time, we can finish \textbf{PlaneDecision}.
\end{proof}

Finding a good placement on $S$ is similar to finding a goal placement on the whole configuration space. $S$ is partitioned into cells by the intersections of event polygons with $S$. We need to find a region on $S$ containing a good placement that does not intersect any event polygons. 

We present a subroutine \textbf{LineDecision} that finds, for a line $l \subset S$, the relative position of the set of good placements on $S$ to $l$.
\begin{proposition}\label{prop_line_decision}
    For a line $l \subset S$, we can perform \textbf{LineDecision} in $O(n)$ time.
\end{proposition}
\begin{figure}[ht]
  \centering 
  \tdplotsetmaincoords{110}{0}
  \begin{tikzpicture}[tdplot_main_coords,thick,scale=1.4,line join=round]
    \begin{scope}
      \coordinate (B) at (1,-1,2);
      \coordinate (C) at (-1.5,-1,0);
      \coordinate (D) at (0,2.5,2);
      
      \coordinate (AC) at (-1,0,2);
      \coordinate (DE) at (0,0.75,0);
      \coordinate (BE) at (0.5,-1,0);
      
      \draw (AC) -- (C) -- (DE) -- (BE) -- (B) -- cycle;
      \draw (B) -- (D) -- (C);
      \draw (AC) -- (D) -- (DE);
      \draw[dotted] (B) -- (C) -- (BE);

      \node at (-0.2,1.5) {$P$};
    \end{scope}
    \begin{scope}[shift={(2.5,-0.5)}]
      \coordinate (A0) at ( 0  ,0,2);
      \coordinate (B0) at ( 0.5,2,2);
      \coordinate (C0) at ( 2  ,0,2);
      \coordinate (A1) at (-0.5,0,0);
      \coordinate (B1) at ( 0  ,2,0);
      \coordinate (C1) at ( 1.5,0,0);
      
      \coordinate (A2) at ( 0.0625,0, 2.25);
      \coordinate (B2) at ( 0.5625,2, 2.25);
      \coordinate (C2) at ( 2.0625,0, 2.25);
      \coordinate (A3) at (-0.5625,0,-0.25);
      \coordinate (B3) at (-0.0625,2,-0.25);
      \coordinate (C3) at ( 1.4375,0,-0.25);
      
      \draw[dotted] (A0) -- (B0) -- (C0) -- cycle;
      \draw[dotted] (A1) -- (B1) -- (C1) -- cycle;
      \draw         (A2) -- (A3);
      \draw         (B2) -- (B3);
      \draw         (C2) -- (C3);
      
      \node at (0.675,2) {$Q+l$};
    \end{scope}
  \end{tikzpicture}
  \caption{The convex polyhedron $I$ is formed by interesecting $P$ and $(Q+l)$.}\label{fig_line_decision}
\end{figure}
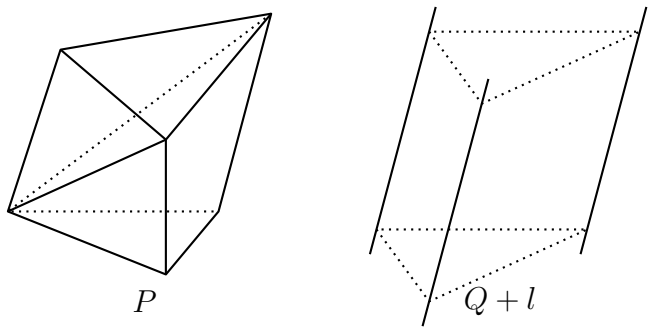
\begin{proof}
    First, we compute $\operatorname{max}_{v \in l} f(v)$ and a vector achieving the maximum. We parameterize the line $l$ by $p+vt$ where $t$ is the parameter and $p, v \in \mathbb{R}^3$. The horizontal cross-section of $I = P\cap(Q+l)$ at height $t$ has area $f(p+vt)$. Since $I$ is the intersection of two convex polytopes with $O(n)$ vertices (see \Cref{fig_line_decision}), Chazelle's algorithm \cite{chazelle1992} computes $I$ in $O(n)$ time. Then, \cite[Theorem 3.2]{avis1996} computes the maximum cross-section in $O(n)$ time.
    
    Now, by the same argument and method as in the proof of \Cref{lemma_side_decision}, we can finish \textbf{LineDecision} in $O(n)$ time. In the case where $\operatorname{max}_{v \in l} f(v) = 0$, we report the side of $l$ containing $S \cap (P - Q)$.
\end{proof}
Whenever our subroutine \textbf{LineDecision} reports a good placement is found on a line, we can let the algorithm terminate. Thus, we can assume it always reports a half-plane of $S$ containing a good placement. 

We now present \textbf{PlaneDecision}. If $S$ is horizontal, then we only need to find the maximum overlap of the convex polygons $P \cap S$ and $Q$ using De Berg et al.'s algorithm \cite{deberg1996}, which takes $O(n \log n)$ time. Thus, we assume $S$ is non-horizontal.

\begin{algorithm}[ht]
    \DontPrintSemicolon
    \SetKwInOut{Input}{input}\SetKwInOut{Output}{output}
    \Input{A plane $S \subset \mathbb{R}^3$}
    \Output{A translation $v \in S$ maximizing the area $|P \cap (Q + v)|$}
    Compute $S \cap (P - Q)$ and set it to be our initial target region.\;
    Locate a strip on $S$ containing a good placement whose interior intersects $O(n)$ event polygons.\;
    Recursively construct a $(1/2)$-cutting of the strip to find a triangle containing a good placement that does not intersect any event polygon\;
    \caption{Pseudocode for \textbf{PlaneDecision}}\label{pseudo_plane}
\end{algorithm}

As in \Cref{pseudo_plane}, we break down \textbf{PlaneDecision} into three steps. We have already explained Step $1$, where we compute $S \cap (P - Q)$, so we begin with Step $2$.

\subsubsection{PlaneDecision: Step 2}
We want to find a strip on $S$ strictly between $z = 0$ and $z = 1$ that intersects $O(n)$ event polygons. Since there are no vertices of $P$ with $z$-coordinate in the interval $(0, 1)$, there are no event polygons of type~\ref{type1} in this range, and we will only need to consider event polygons of type~\ref{type2} and type~\ref{type3}. 

We look at the intersection points of $S$ with the edges of the event polygons. These edges come from the set $\{e_i - v_j|e_i \text{ non-horizontal edge of }P, \, v_j \text{ vertex of } Q\}$. Without loss of generality, assume that $S$ is parallel to the $y$-axis. We are interested in the $z$-coordinates of the intersections, so we project everything into the $xz$-plane. Then, $S$ becomes a line, which we denote by $l_S$, and each edge $e_i - v_j$ becomes a segment whose endpoints lie on $z = 0$ and $z = 1$. Suppose each edge $e_i$ projects to a segment $s_i$, and each $v_j$ projects to a point $x_j$ on the $x$-axis. Then, we get $O(n^2)$ segments $s_i - x_j$ with endpoints on $z = 0$ and $z = 1$, and the line $l_S$ that intersect them in some places. 
\begin{lemma}\label{lemma_plane_step_2.1}
    In $O(n \log n)$ time, we can locate a strip $R = \{(x, y, z)\in S | z \in [z_0, z_1]\}$ whose interior contains a good placement and intersects none of the edges of the event polygons.
\end{lemma}
\begin{proof}
    By our setup, we want to find a segment on $l_S$ whose interior does not intersect any segment of the form $s_i - x_j$. 
    
    Since $s_i$ are projections of edges of a convex polyhedron, we can separate them into two sets such that edges from the same set do not intersect (we take the segments that are projections of the edges of the ``front'' side and ``back'' side, respectively), allowing the two extremal edges to appear in both sets. We will process each set separately. This can be done by identifying the extremal points points on the top and bottom faces of $P$ in the $x$ direction, which can be done in $O(\log n)$ time. 
    
    For a set of non-intersecting segments, since they all have endpoints on the line $z = 0$ and $z = 1$, we can sort them by the sum of the $x$-coordinates of their two endpoints. This takes $O(n \log n)$ time. We further separate these segments into two sets by slope: those that make a smaller angle than $l_S$ with the positive $x$-axis, and those that make a larger angle. 
    
    Suppose we now have a set of non-intersecting segments that all make larger angles than $l_S$ with the positive $x$-axis, $s_1, s_2, \ldots, s_m$, where $m = O(n)$. We also sort the projections of the vertices of $Q$, $x_1, \ldots, x_q$, in decreasing order by $x$-coordinate. This can be done in $O(\log n)$ time by identifying the extremal vertices of $Q$ in the $x$-direction.
    
    Let $z_{ij}$ be the $z$-coordinate of the intersection of the line containing $s_i - x_j$ with $l_S$. Let $M$ be an $m\times q$ matrix with $(i,j)$-th entry given by
    \[
        M_{ij} = 
        \begin{cases}
        0 & z_{ij} \leq 0 \\
        z_{ij} & z_{ij} \in (0,1) \\
        1 & z_{ij} \geq 1
        \end{cases}.
    \]
    We claim that $M$ is a sorted matrix. To see this, consider any fixed row $r$ and indices $i < j$. Then the line containing $s_r - x_i$ lies strictly to the left of the line containing $s_r - x_j$ since $x_i > x_j$. This means that $z_{ri} < z_{rj}$. Thus, every row of $M$ is in increasing order. Similarly, for a fixed column $c$ and indices $i<j$, the segment $s_i - x_c$ lies strictly to the left of the segment $s_j - x_c$. Then, if they both intersect $l_S$, we must have $z_{ic}<z_{jc}$. If $s_i - x_c$ does not intersect $l_S$ and $s_j - x_c$ does, then $s_i - x_c$ must lie on the left of $l_S$ and thus $M_{ic} = a < z_{jc} = M_{jc}$. Similarly, if $s_i - x_c$ intersects $l_S$ and $s_j - x_c$ does not, then $s_j - x_c$ must lie on the right of $l_S$ and thus $M_{ic} = z_{ic} < b = M_{jc}$. If they both do not intersect $l_S$, then still $M_{ic} \leq M_{jc}$ since it is impossible to have $M_{ic} = b$ and $M_{jc} = a$. This proves our claim. 
    
    By \Cref{lemma_sorted_matrix}, we can find the $k$-th smallest value in $M$ in $O(m + q) = O(n)$ time. Thus, we can perform a binary search on these $z$-coordinates of the intersections of the edges $e_i - v_j$ with $S$. Each time we perform a \textbf{LineDecision} on the line with the median $z$-coordinate of the remaining entries to eliminate half of the intersections. After $O(\log n)$ iterations or $O(n \log n)$ time, we find a strip on $S$ containing a good placement that contains no intersections with this group of edges. 
    
    We repeat the same procedure for the other three groups and compute the intersection of the four strips to find a strip containing a good placement that contains no intersections with any edge of the event polygons.
\end{proof}

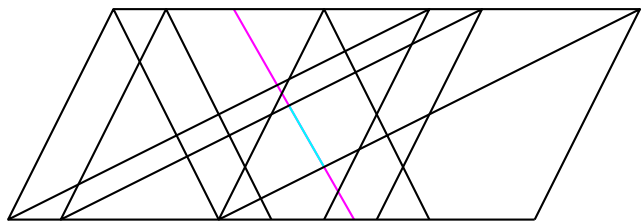
\begin{figure}[ht]
  \centering 
  \begin{tikzpicture}[thick,scale=1.4,line join=round]
    \draw [color=botColor] (23/7,0) -- (15/7,2);
    \draw [color=topColor] (3,1/2) -- (8/3,13/12);
    \foreach \vQx in {0,0.5,2}{
      \begin{scope}[shift={(\vQx,0)}]
        \coordinate (A) at (0,0);
        \coordinate (B) at (1,2);
        \coordinate (C) at (2,0);
        \coordinate (D) at (3,0);
        \coordinate (E) at (4,2);
        
        \draw (D) -- (E) -- (A) -- (B) -- (C);
      \end{scope}}
    \draw (0,0) -- (5,0);
    \draw (1,2) -- (6,2);
  \end{tikzpicture}
  \caption{Projecting the configuration space onto the $xz$-plane. The projection of $S$ is the magenta line segment, and the projection of the strip $R$ obtained form \Cref{lemma_plane_step_2.1} is the cyan line segment.}\label{fig_plane_decision}
\end{figure}

Our current target region, the strip $R$ we obtained from \Cref{lemma_plane_step_2.1} (see \Cref{fig_plane_decision}), intersects few event polygons and we can compute them efficiently.

\begin{lemma}\label{lemma_plane_step_2.2}
    The interior of the region $R$ intersects $O(n)$ event polygons, and we can compute them in $O(n \log n)$ time. 
\end{lemma}
\begin{proof}
    For a vertex $v$ of $Q$, it contributes the $O(n)$ event polygons of type~\ref{type2} that are the faces of $P - v$. The intersection of the boundary of $P - v$ with $S$ is a convex polygon. Since there are no intersections with edges of event polygons inside the strip $R$, at most two edges of the convex polygon can lie inside $R$, one on the ``front side'' and the other on the ``back side.'' 
    
    To compute these two segments on $R$, we first consider the two sorted matrices given in the proof of \Cref{lemma_plane_step_2.1} that together describe the edges on the ``front side'' and look at the column associated to $-v$. We find, for each column, the two (or zero) adjacent entries that contain the $z$-coordinates of $R$ in between. The two of the at most four that are closest to the strip will be the endpoints of the segment that intersect the strip on the ``front side.'' Computing this segment takes $O(\log n)$ time since we can use binary search on the columns to find the desired entries. We do the same to find the segment on the ``back side.'' We do this for all vertices of $Q$ to find the $O(n)$ intersections with the event polygons of type~\ref{type2} in $O(n \log n)$ time.
    
    For an edge $e$ of $P$, it contributes $O(n)$ event polygons of type~\ref{type3} that form the surrounding sides of a ``cylinder'' with base congruent to $-Q$. Again, each of these ``cylinders'' intersect the strip $R$ in at most two faces, so there are $O(n)$ intersections of $R$ with event polygons of type~\ref{type3}. We can compute these segments by performing the binary search on the row of one of the sorted matrices associated to the edge $e$. The two entries immediately below the strip and the two immediately above the strip define the at most two segments intersecting $R$. Similar to the procedure above, this takes $O(\log n)$ time for each edge of $P$, thus $O(n \log n)$ time in total. 
\end{proof}

\subsubsection{PlaneDecision: Step 3}
Now, we have a target region $R$ and the $O(n)$ intersections it makes with the event polygons. 
\begin{lemma}\label{lemma_plane_step_3}
    In $O(n \log n)$ time, we can find a region $R' \subset R$ containing a good placement that does not intersect any of the $O(n)$ event polygons.
\end{lemma}
\begin{proof}
    We recursively construct a $(1/2)$-cutting of the target region. By \Cref{lemma_cutting}, a  $(1/2)$-cutting of constant size can be computed in $O(n)$ time. We perform \textbf{LineDecision} on the lines of the cutting to decide on which triangle to recurse. After $O(\log n)$ iterations, we have a target region $R'$ that intersects no event polygons. This procedure runs in $O(n \log n)$ time.
\end{proof}

Finally, since the overlap function is quadratic on our final region $R'$, we can solve for the maximum using standard calculus. After finding $\operatorname{max}_{v \in S} f(v)$ and a vector achieving it $O(n \log n)$ time, by \Cref{lemma_side_decision}, we can perform \textbf{PlaneDecision} on $S$ in the same time bound.
\begin{proposition}\label{prop_max_on_plane}
    For a plane $S$, we can perform \textbf{PlaneDecision} in $O(n \log n)$ time. 
\end{proposition}

\subsection{Stage 2}
With the general \textbf{PlaneDecision} at our disposal, we now move on to State $2$, the main component of our algorithm. We project the entire configuration space and the event polygons onto the $xz$-plane in order to find a target region $D$ whose preimage $D + l_y$ intersects few event polygons, where $l_y$ is the $y$-axis (see \Cref{fig_stage_2}). 

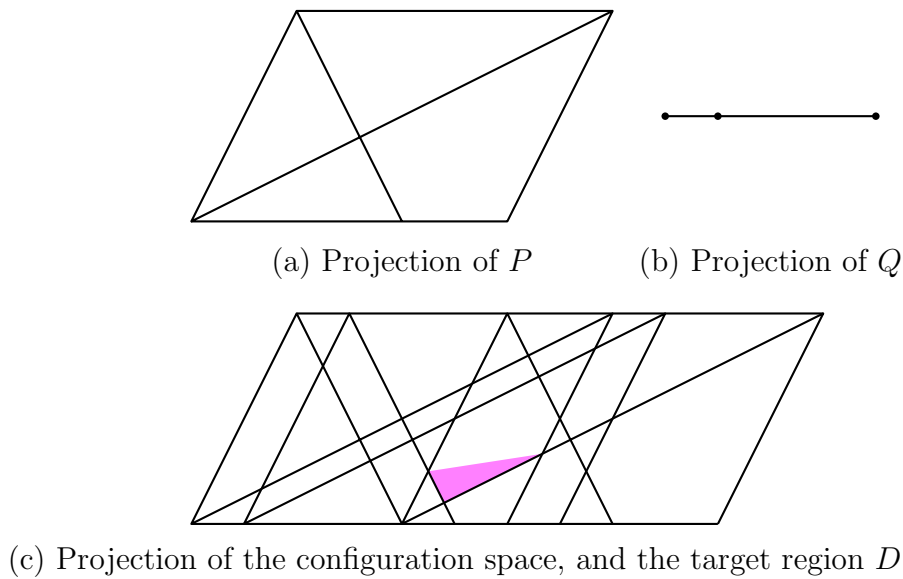
\begin{figure}[ht]
  \centering 
  \begin{tikzpicture}[thick,scale=1.4,line join=round]
    \begin{scope}
      \coordinate (A) at (0,0);
      \coordinate (B) at (1,2);
      \coordinate (C) at (2,0);
      \coordinate (D) at (3,0);
      \coordinate (E) at (4,2);

      \draw (B) -- (E);
      \draw (A) -- (D) -- (E) -- (A) -- (B) -- (C);
      
      \node at (2,-0.4) {(a) Projection of $P$};
    \end{scope}
    \begin{scope}[shift={(4.5,1)}]
      \draw[draw=none] (1,1) -- (1,-1);
      \draw (0,0) node[circle,fill,inner sep=1pt]{}
      -- (0.5,0) node[circle,fill,inner sep=1pt]{}
      -- (2,0) node[circle,fill,inner sep=1pt]{};
      
      \node at (1,-1.4) {(b) Projection of $Q$};
    \end{scope}
  \begin{scope}[shift={(0,-2.875)}]
    \draw [draw=none, fill=botColor, fill opacity=0.5]
    (2.25,0.5) -- (2.4,0.2) -- (10/3,2/3) -- cycle;
    \foreach \vQx in {0,0.5,2}{
      \begin{scope}[shift={(\vQx,0)}]
        \coordinate (A) at (0,0);
        \coordinate (B) at (1,2);
        \coordinate (C) at (2,0);
        \coordinate (D) at (3,0);
        \coordinate (E) at (4,2);
        
        \draw (D) -- (E) -- (A) -- (B) -- (C);
      \end{scope}}
    \draw (0,0) -- (5,0);
    \draw (1,2) -- (6,2);
      \node at (2.5,-0.35) {(c) Projection of the configuration space, and the target region $D$};
  \end{scope}
  \end{tikzpicture}
  \caption{Projecting onto the xz-plane.}\label{fig_stage_2}
\end{figure}

The non-horizontal edges of the event polygons project to segments on the strip $0<z<1$ on the $xz$-plane. We characterize our desired region $D$ in the following lemma.
\begin{lemma}\label{lemma_stage_2.1}
    For a region $D$ that does not intersect any of the segments that are the projections of the non-horizontal edges of the event polygons, the preimage $D + l_y$ intersects $O(n)$ event polygons.
\end{lemma}
\begin{proof}
    For any region $D$ on the $xz$-plane, the set of event polygons that the ``tube'' $D + l_y$ intersects is precisely the set of projected event polygons that $D$ intersects. Now, let $D$ be a region that does not intersect any segment from the projections of the event polygons.
    
    Let $s_1, s_2, \ldots, s_m$ be the segments that are the projections of the non-horizontal edges of $P$, and let $x_1, \ldots, x_q$ be the projections of the vertices of $Q$ on the $x$-axis and assume that they are sorted by decreasing $x$-coordinate. Then, the projections of the non-horizontal edges of the event polygons are precisely $s_i - x_j$. 
    
    We first split the segments into four groups. Let $s_1, \ldots, s_{m_1}$ be the projections of the non-horizontal edges of $P$ on the ``front side,'' and $s_{m_1+1}, \ldots, s_{m}$ be those on the ``back side.'' The at most two edges visible on both the front and the back may be repeated. Then the segments from either group are pairwise non-intersecting. Similarly, we split the vertices of $Q$ into a front side and a backside, including the at most two vertices visible on both the front and back in both sets. We consider the segments in the configuration space made by one of the two groups of edges of $P$ and one of the two groups of vertices of $Q$. The other three sets of segments are processed similarly. 
    
    Suppose that the segments we consider are $s_1, \ldots, s_{m_1}$, and the projected vertices are $x_1, \ldots, x_{q_1}$. Suppose the segments are sorted by increasing sum of the $x$-coordinates of their endpoints, and that the vertices are sorted by decreasing $x$-coordinate. The event polygons of type~\ref{type2} are the trapezoids or triangles between segments $s_i - x_j$ and $s_{i+1} - x_j$ for each of the four groups of segments. For each fixed projected vertex $x$, the region $D$ intersects at most one event polygon of type~\ref{type2} for each group. Thus, $D$ intersects $O(n)$ event polygons of type~\ref{type2}. Similarly, the event polygons of type~\ref{type3} are the parallelograms between segments $s_i - x_j$ and $s_i - x_{j+1}$ for each of the four groups of segments. For each fixed segment $s_i$, $D$ intersects at most one event polygon of type~\ref{type3}, thus it intersects $O(n)$ event polygons of type~\ref{type3} in total. 
\end{proof}

Now it remains to efficiently find such a region $D$ with $D + l_y$ containing a goal placement and compute the $O(n)$ event polygons that intersect its interior.
\begin{lemma}\label{lemma_stage_2.2}
    In $O(n \log^2 n)$ time, we can find a triangle $D$ in the $xz$-plane such that the interior of $D + l_y$ contains a goal placement and intersects $O(n)$ event polygons. We can compute these $O(n)$ event polygons in the same time bound.
\end{lemma}
\begin{proof}
    The computation of $D$ is a direct application of \Cref{thm_prune_and_search}, where $m = O(n)$. Calling the oracle on a line $l$ in the $xz$-plane is running the \textbf{PlaneDecision} algorithm on the plane parallel to the $y$-axis that projects to $l$. We compute a triangle for each of the four groups of segments, take their intersection, and triangulate the intersection using $O(1)$ calls to \textbf{PlaneDecision}. Thus, we can compute the desired triangle $D$ in $O(n \log^2 n)$ time.
    
    To compute the event polygons intersecting the interior of $D + l_y$ is simple, since we have shown in the proof of \Cref{lemma_stage_2.1} that $D$ intersects at most one projection of an event polygon of each type in each of the four groups for a fixed vertex $x_j$ (for type~\ref{type2}) or segment $s_i$ (for type~\ref{type3}). Once we have $D$, we can compute these polygons by binary search on each of the $O(n)$ groups of $O(n)$ non-intersecting segments to find the two between which $R$ lies. Also, the event polygons all have constant complexity so computing all of them takes linear time. We can recover the event polygons from their projections and compute the planes that contain them in linear time. Thus, this entire process can be done in $O(n \log n)$ time.
\end{proof}

\subsection{Stage 3}
Now, we have a target region $R = D+l_y$ whose interior contains a goal placement, and we have the $O(n)$ event polygons that intersect it. 
\begin{lemma}\label{lemma_stage_3}
    In $O(n \log^2 n)$ time, we can find a region $R' \subset R$ containing a goal placement that does not intersect any of the $O(n)$ event polygons.
\end{lemma}
\begin{proof}
    We recursively construct a $(1/2)$-cutting of the target region. By \Cref{lemma_cutting}, a $(1/2)$-cutting of constant size can be computed in $O(n)$ time. We perform \textbf{PlaneDecision} on the planes of the cutting to decide on which simplex to recurse. After $O(\log n)$ iterations, we have a target region $R'$ that intersects no event polygons. This procedure runs in $O(n \log^2 n)$ time.
\end{proof}

Finally, since the overlap function is quadratic on our final region $R'$, we can solve for the maximum using standard calculus. This concludes the proof of \Cref{thm_algo}.

\section{Maximum overlap of three convex polygons}\label{sec_three_polygons}
Let $P$, $Q$, $R$ be three convex polygons with $n$ vertices in total in the plane. We want to find a pair of translations $(v_{Q}, v_{R}) \in \mathbb{R}^4$ that maximizes the overlap area $g(v_{Q}, v_{R}) = |P \cap (Q + v_{Q}) \cap (R + v_{R})|$. 

In this problem, the configuration space is four-dimensional. An easy extension of \Cref{prop_concavity} and \Cref{cor_unimodality} shows that the function of overlap area is again unimodal. This time, we have four-dimensional \textit{event polyhedra} instead of event polygons that divide the configuration space into four-dimensional cells on which $g(v_{Q}, v_{R})$ is quadratic. We call a hyperplane containing an event polyhedron an \textit{event hyperplane}, and they are defined by two types of events:
\begin{enumerate}[label = (\Roman*)]
    \item \label{hypertype1} When one vertex of $P$, $Q + v_Q$ or $R + v_R$ lies on an edge of another polygon. There are $O(n)$ groups of $O(n)$ parallel event hyperplanes of this type.
    \item \label{hypertype2} When an edge from each of the three polygons intersect at one point. There are $O(n^3)$ event hyperplanes of this type. 
\end{enumerate}

To overcome the difficulty of dealing with the $O(n^3)$ event hyperplanes of type~\ref{hypertype2}, we first prune the configuration space to a region intersecting no event hyperplanes of type~\ref{hypertype1}. We then show that the resulting region only intersects $O(n)$ event hyperplanes of type~\ref{hypertype2}. 

Similar to \Cref{thm_algo}, we want an algorithm \textbf{HyperplaneDecision} that computes, for a hyperplane $H \subset \mathbb{R}^4$, the maximum $\operatorname{max}_{(v_Q,v_R) \in H} g(v_Q, v_R)$ and the relative location of the goal placement to $H$. In fact, we will only need to perform \textbf{HyperplaneDecision} on some hyperplanes.

\begin{proposition}\label{prop_hyperplane_decision}
    Suppose $H$ is a hyperplane that satisfies one of the following three conditions:
    \begin{enumerate}[label = (\arabic*)]
        \item $H$ is orthogonal to a vector $(x_1, y_1, 0, 0)$ for some $x_1, y_1 \in \mathbb{R}$.
        \item $H$ is orthogonal to a vector $(0, 0, x_2, y_2)$ for some $x_2, y_2 \in \mathbb{R}$.
        \item $H$ is orthogonal to a vector $(x_1, y_1, -x_1, -y_1)$ for some $x_1, y_1 \in \mathbb{R}$.
    \end{enumerate}
    Then, we can perform \textbf{HyperplaneDecision} on $H$ in $O(n \log^2 n)$ time. 
\end{proposition}
\begin{proof}
    We provide the algorithm for $H$ orthogonal to $(x_1, y_1, 0, 0)$ for some $x_1, y_1 \in \mathbb{R}$, and the other two types follow similarly. 
    
    We reinterpret the problem of finding $\operatorname{max}_{(v_Q,v_R) \in H} g(v_Q, v_R)$ as a polyhedron-polygon matching problem. In $H$, we allow $R$ to move freely, and $Q$ moves in a line $l$ perpendicular to $(x_1, y_1)$. We parameterize $l$ by $l = p + vt$, and form the convex polyhedron (see \Cref{fig_ipq})
    \[
    I_{PQ} = \{(x, y, t)| (x, y) \in P\} \cap \{(x, y, t)| (x, y) \in (Q + p + vt)\}.
    \] 
    By \cite{chazelle1992}, $I$ can be computed in $O(n)$ time. In addition, the cross-section of $I$ at $t = t_0$ is $P \cap (Q + p + vt)$. Then, we see that finding $\operatorname{max}_{(v_Q, v_R) \in H} g(v_Q, v_R)$ is the same as finding a translation maximizing the intersection of $I$ and $R$. By \Cref{thm_algo}, this can be done in $O(n \log^2 n)$ time. 

    Using the formal perturbation argument in \Cref{lemma_side_decision}, \textbf{HyperplaneDecision} on $H$ can be completed in the same time bound.

\begin{figure}[ht]
  \centering 
  \tdplotsetmaincoords{110}{0}
  \begin{tikzpicture}[tdplot_main_coords,thick,scale=1.4,line join=round]
    \begin{scope}[shift={(-1,0)}]
      \coordinate (A0) at (0  ,0  ,0);
      \coordinate (B0) at (1.5,0  ,0);
      \coordinate (C0) at (2  ,1.5,0);
      \coordinate (D0) at (0.5,1.5,0);
      \coordinate (A1) at (0  ,0  ,2);
      \coordinate (B1) at (1.5,0  ,2);
      \coordinate (C1) at (2  ,1.5,2);
      \coordinate (D1) at (0.5,1.5,2);
      \coordinate (A2) at (0  ,0  ,-0.5);
      \coordinate (B2) at (1.5,0  ,-0.5);
      \coordinate (C2) at (2  ,1.5,-0.5);
      \coordinate (D2) at (0.5,1.5,-0.5);
      \coordinate (A3) at (0  ,0  ,2.5);
      \coordinate (B3) at (1.5,0  ,2.5);
      \coordinate (C3) at (2  ,1.5,2.5);
      \coordinate (D3) at (0.5,1.5,2.5);
      
      \draw[dotted] (A0) -- (B0) -- (C0) -- (D0) -- cycle;
      \draw[dotted] (A1) -- (B1) -- (C1) -- (D1) -- cycle;
      \draw         (A2) -- (A3);
      \draw[dotted] (B2) -- (B3);
      \draw         (C2) -- (C3);
      \draw         (D2) -- (D3);

      \node at (1.25,2.5) {$P+\text{($z$-axis)}$};
    \end{scope}
    \begin{scope}[shift={(3.5,0)}]
      \coordinate (A0) at ( 0  ,0,2);
      \coordinate (B0) at ( 0.5,2,2);
      \coordinate (C0) at ( 2  ,0,2);
      \coordinate (A1) at (-0.5,0,0);
      \coordinate (B1) at ( 0  ,2,0);
      \coordinate (C1) at ( 1.5,0,0);
      
      \coordinate (A2) at ( 0.0625,0, 2.25);
      \coordinate (B2) at ( 0.5625,2, 2.25);
      \coordinate (C2) at ( 2.0625,0, 2.25);
      \coordinate (A3) at (-0.5625,0,-0.25);
      \coordinate (B3) at (-0.0625,2,-0.25);
      \coordinate (C3) at ( 1.4375,0,-0.25);
      
      \draw[dotted] (A0) -- (B0) -- (C0) -- cycle;
      \draw[dotted] (A1) -- (B1) -- (C1) -- cycle;
      \draw         (A2) -- (A3);
      \draw         (B2) -- (B3);
      \draw         (C2) -- (C3);
      
      \node at (0.675,2.5) {$Q + l$};
    \end{scope}
  \end{tikzpicture}
  \caption{The convex polyhedron $I_{PQ}$ is the intersection of these two objects.} \label{fig_ipq}
\end{figure}
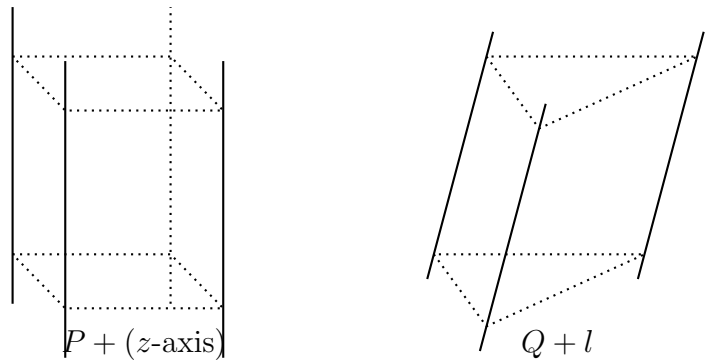
\end{proof}

Using \Cref{prop_hyperplane_decision}, we can prune the configuration space to a region that intersects no event hyperplanes of type~\ref{hypertype1} and $O(n)$ event hyperplanes of type~\ref{hypertype2}. 

\begin{proposition}\label{prop_Tpqr}
    We can compute a 4-polytope $T_{PQR}$ of complexity $O(1)$ in $O(n\log^3 n)$ time such that
    \begin{enumerate}[label = (\arabic*)]
        \item the goal placement lies on $T_{PQR}$,
        \item no hyperplane of type~\ref{hypertype1} intersects the interior of $T_{PQR}$, and
        \item only $O(n)$ event polyhedrons of type~\ref{hypertype2} passes through $T_{PQR}$.
    \end{enumerate}
    The hyperplanes of type~\ref{hypertype2} intersecting the interior of $T_{PQR}$ are obtained in the same time bound. Furthermore, the 3-tuples of edges of $P$, $Q$ and $R$ defining the hyperplanes are also obtained in the same time bound.
\end{proposition}
\begin{proof}
    If a \textbf{HyperplaneDecision} reports a goal placement, we are done. Thus, we assume that \textbf{HyperplaneDecision} always reports a halfspace containing a goal placement. 
    
    Each event hyperplane containing an event polyhedron of a vertex of $P$ on an edge of $Q + v_Q$ or an event polyhedron of a vertex of $Q + v_Q$ on an edge of $P$ is orthogonal to some $(x_1, y_1, 0, 0)$. We project all these event hyperplanes into the $2$-flat $S_{PQ} = \{(x_1, y_1, 0, 0) | x_1, y_1 \in \mathbb{R}\}$. Then, the images are $O(n)$ groups of $O(n)$ parallel lines. We can therefore apply \Cref{thm_prune_and_search} to these lines, where an oracle call on a line $l$ is running \textbf{HyperplaneDecision} on the hyperplane that projects to $l$ on $S_{PQ}$, which is orthogonal to some $(x_1, y_1, 0, 0)$. Thus, by \Cref{prop_hyperplane_decision}, we can find a triangle $T_{PQ} \subset S_{PQ}$ whose interior does not intersect any event hyperplane as described above in $O(n\log^3 n)$ time.
    
    Similarly, we can find the triangles 
    \[
    T_{PR} \subset \{(0, 0, x_2, y_2) | x_2, y_2 \in \mathbb{R}\} \quad \text{and} \quad T_{QR} \subset \{(x_1, y_1, -x_1, -y_1) | x_1, y_1 \in \mathbb{R}\}
    \] 
    corresponding to the other event hyperplanes of type~\ref{hypertype1} in $O(n\log^3 n)$ time. Then, the interior of 
    \begin{align*}
        T_{PQR} = \{(x_1, y_1, x_2, y_2) | & (x_1, y_1, 0, 0) \in T_{PQ}, \, (0, 0, x_2, y_2) \in T_{PR}, \,  \\
        & \left(\frac{x_1 - x_2}{2}, \frac{y_1 - y_2}{2}, \frac{x_2 - x_1}{2}, \frac{y_2 - y_1}{2} \right) \in T_{QR}\}
    \end{align*}
    does not intersect any event hyperplane of type~\ref{hypertype1} and contains a goal placement.
    
    Since the interior of $T_{PQR}$ intersects no event hyperplane of type~\ref{hypertype1}, the pairwise configuration of $P$ and $Q$, $P$ and $R$, $Q$ and $R$ are fixed (the pairwise edge incidences are fixed). Since any edge $e_P$ of $P$ intersects at most two edges of $Q$ and at most two edges of $R$ inside $T_{PQR}$, there are at most four event hyperplanes of type~\ref{hypertype2} where $e_P$ is concurrent with an edge of $Q$ and an edge of $R$. Thus, at most $4n$ event hyperplanes of type~\ref{hypertype2} intersect the interior of $T_{PQR}$. 
\end{proof}

In the rest of the section, we fix $T_{PQR}$ as in \Cref{prop_Tpqr}. Moreover, let 
\[ 
  f(v_P,v_Q) = \begin{cases}
    |P \cap (Q + v_Q) \cap (R + v_R)| & \text{if } (v_Q, v_R)\in T_{PQR} \\
    0 & \text{otherwise.}
  \end{cases}
\]

\begin{proposition}\label{prop_suppf}
  Let $S$ be any $m$-flat in the configuration space. In $O(n)$ time, we can find a point in $S \cap \mathrm{supp\,} f$, or report that $S \cap \mathrm{supp\,} f$ is empty.
\end{proposition}
\begin{proof}
  Notice that $\mathrm{supp\,} f$ is a convex 4-polytope whose face are hyperplanes of type I or type II. Let $H$ be a hyperplane of type II intersecting the interior of $T_{PQR}$. Then $H$ contains a face of $\mathrm{supp\,} f$ if and only if a polygon $P \cap Q$ is tangent to $R$ in $H \cap T_{PQR}$. This can be tested in constant time, so we can find all faces of $\mathrm{supp\,} f$ in $O(n)$ time. Our problem become a feasibility test of a linear programming of size $O(n)$, which can be solved in $O(n)$ time by Megiddo's algorithm \cite{megiddo1984}.
\end{proof}
%\threepolys*
\begin{proof}[Proof of \Cref{thm_three_polygons}]
    Take $T_{PQR}$ as in \Cref{prop_Tpqr}. Let
    \[ 
        f(v_P,v_Q) = \begin{cases}
            |P \cap (Q + v_Q) \cap (R + v_R)| & \text{if } (v_Q, v_R)\in T_{PQR} \\
            0 & \text{otherwise.}
        \end{cases}
    \]
    Then $f$ is unimodal and the maximum of $f$ is the goal placement. Given an $m$-flat $S$, we want to compute the maximum of $f$ on $S$ in $O(n \log^{m-1})$ time by induction on $m \in \{1, 2, 3, 4\}$.

    If $m = 1$, this can be done in $O(n)$ time by \Cref{prop_line_decision}. Assume that $m > 1$. Then $S \cap T_{PQR}$ can be computed in $O(1)$ time. Given an $(m-1)$-flat $l \subset S$, we can use \Cref{prop_suppf} and the perturbation method as in \Cref{lemma_side_decision} to report the relative position of the maximum over $S$. There are $O(n)$ event hyperplane intersecting $S\cap T_{PQR}$. Thus, by \Cref{lemma_cutting}, we can recursively construct $(1/2)$-cuttings to give an $O(n \log^{m-1})$ time algorithm to find the maximum of $f$ on $S$.
\end{proof}

% With \Cref{prop_Tpqr}, we can finish our algorithm.
% \threepolys*
% \begin{proof}
    % Take $T_{PQR}$ as in \Cref{prop_Tpqr}. Let
    % \[ 
        % f(v_P,v_Q) = \begin{cases}
            % |P \cap (Q + v_Q) \cap (R + v_R)| & \text{if } (v_Q, v_R)\in T_{PQR} \\
            % 0 & \text{otherwise.}
        % \end{cases}
    % \]
    % Then $f$ is unimodal and the maximum of $f$ is the goal placement. Given an $m$-flat $S$, we want to compute the % maximum of $f$ on $S$ in $O(n \log^{m-1})$ time by induction on $m \in \{1, 2, 3, 4\}$.

    % If $m = 1$, this can be done in $O(n)$ time by \Cref{prop_line_decision}. Assume that $m > 1$. Then $S \cap T_{PQR}$ % can be computed in $O(1)$ time. Given an $(m-1)$-flat $l \subset S$, we can use the perturbation method as in \Cref{lemma_side_decision} to report the relative position of the maximum over $S$. There are $O(n)$ event hyperplane intersecting $S\cap T_{PQR}$. Thus, by \Cref{lemma_cutting}, we can recursively construct $(1/2)$-cuttings to give an $O(n \log^{m-1})$ time algorithm to find the maximum of $f$ on $S$.
% \end{proof}

\section{Minimum symmetric difference of two convex polygons under homothety}\label{sec_symmetric_difference}
A homothety $\varphi\colon\mathbb{R}^2\rightarrow\mathbb{R}^2$ is a composition of a scaling and a translation. Let $\lambda>0$ be the scaling factor and $v$ be the translation vector of $\varphi$. Then
\[\varphi(A)=\lambda A + v = \{\lambda p + v \mid p\in A\}.\]
Define the \textit{symmetric difference} of sets $A, B \subset \mathbb{R}^2$ to be 
\begin{align*}
    A \triangle B := & (A \cup B) \setminus (A \cap B) \\
    = & (A \setminus B) \cup (B \setminus A).
\end{align*}

Let $P$ and $Q$ be convex polygons with $n$ vertices in total. We want to find a homothety $\varphi$ of $Q$ that minimizes the area of symmetric difference 
\[
    h(\varphi) = h(x,y, \lambda) = |P \triangle \varphi(Q)|,
\]
where $\varphi(Q) = \lambda Q + (x,y)$. 

Yon et al. \cite{yon2016} consider a slightly more general problem, where they minimize the function 
\[
    h(\varphi) = (2 - 2\kappa)|P \setminus \varphi(Q)| + 2\kappa |\varphi(Q) \setminus P|,
\]
where $\kappa \in (0,1)$ is some constant. When $\kappa = 1/2$, this is the area of symmetric difference function. They give a randomized algorithm that solves this problem in $O(n \log^3 n)$ expected time. We present a faster determinisitc algorithm by relating this problem to the polyhedron-polygon matching problem and then applying a modified version of \Cref{thm_algo}. 

As in \cite{yon2016}, we rewrite the objective function $h(\varphi)$:
\begin{align*}
    h(\varphi) 
    &= 2(1 - \kappa)|P| + 2\kappa|\varphi(Q)| - 2|P \cap \varphi(Q)|\\
    &= 2(1 - \kappa)|P| + 2\kappa |Q| \lambda^2 - 2 |P \cap \varphi(Q)|.
\end{align*}
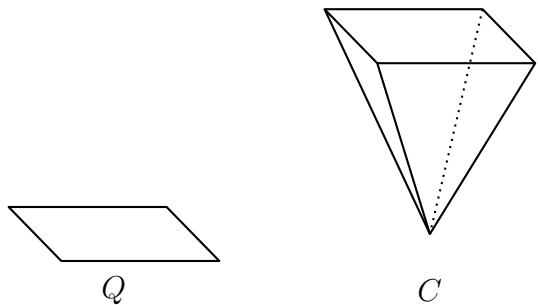
\begin{figure}[ht]
  \centering 
  \tdplotsetmaincoords{110}{0}
  \begin{tikzpicture}[tdplot_main_coords,thick,scale=1.4,line join=round]
    \begin{scope}
      \coordinate (A) at (0  ,0  );
      \coordinate (B) at (1.5,0  );
      \coordinate (C) at (2  ,1.5);
      \coordinate (D) at (0.5,1.5);
      
      \draw (A) -- (B) -- (C) -- (D) -- cycle;

      \node at (1,2.3) {$Q$};
    \end{scope}
    \begin{scope}[shift={(3,0)}]
      \coordinate (O) at (1  ,0.75,0);
      \coordinate (A) at (0  ,0   ,2);
      \coordinate (B) at (1.5,0   ,2);
      \coordinate (C) at (2  ,1.5 ,2);
      \coordinate (D) at (0.5,1.5 ,2);
      
      \draw (A) -- (B) -- (C);
      \draw[dotted] (B) -- (O);
      \draw (C) -- (D) -- (A) -- (O) -- cycle;
      \draw (D) -- (O);

      \node at (1,2.3) {$C$};
    \end{scope}
  \end{tikzpicture}
  \caption{Formation of the cone $C$.}\label{fig_cone}
\end{figure}
Thus, minimizing $h(\varphi)$ is the same as maximizing $f(\varphi) = |P \cap \varphi(Q)| - c \lambda^2$, where $c = \kappa |Q|$. Consider the cone $C = \{(x, y, \lambda)| \lambda \in [0, M], (x, y) \in \lambda Q\}$, where $M = \sqrt{|P|/c}$ (see \Cref{fig_cone}).  Then $f$ is negative for $\lambda > M$ so it is never maximized. We also put $P$ into $\mathbb{R}^3$ by $P = \{(x, y, 0) | (x, y) \in P\}$. Since $f(x, y, \lambda) = |C \cap (P + (-x, -y, \lambda))| - c \lambda^2$, the problem reduces to maximizing the overlap area of the cone $C$ and $P$ under translation subtracted by a quadratic function. To show that we can still use a divide-and-conquer strategy, we identify a region where $f$ is strictly unimodal.
\begin{lemma}[\cite{yon2016}]\label{lemma_region_D}
    The closure $\mathcal{D}$ of the set $\{\varphi \in \mathbb{R}^3| f(\varphi) > 0\}$ is convex. Furthermore, $f(x, y, \lambda)$ is strictly unimodal on $\mathcal{D}$.
\end{lemma}
\begin{proof}
    This follows from \cite[Lemma 2.2]{yon2016} and \cite[Lemma 2.7]{yon2016}. 
\end{proof}
Although it is difficult to directly compute $\mathcal{D}$, note that $-P \subset \mathcal{D}$. With this observation, we show that we can still find the relative position of the set of goal placements to certain planes $S$ in $O(n \log n)$ time with some modifications to 
\textbf{LineDecision} and \textbf{PlaneDecision}. 
\begin{lemma}\label{lemma_modified_line_max}
    For any $l \subset \mathbb{R}^3$, we can compute $\operatorname{max}_{\varphi \in l} f(\varphi)$ or report it is a negative number in $O(n)$ time.
\end{lemma}
\begin{proof}
    If $l$ is horizontal, then we can apply \Cref{prop_line_decision} since $c \lambda$ is constant. Otherwise, we parameterize $l$ by $l = p + v t$ and construct the convex polyhedron $I$ whose cross-section $I(t_0)$ at $t = t_0$ has area $|C \cap (P + (p + vt_0))|$ as in the proof of \Cref{prop_line_decision}. It comes down to maximizing $|I(t)| - c(\lambda(t))^2$, where $\lambda(t)$ is the $\lambda$-coordinate of $p + vt$. Since $\sqrt{|I(t)|}$ is a concave function, $\sqrt{|I(t)|} - \sqrt{c} \lambda(t)$ is also concave, and has the same complexity as $\sqrt{|I(t)|}$. Thus, we can apply \cite[Theorem 3.2]{avis1996} to find the maximum of $\sqrt{|I(t)|} - \sqrt{c} \lambda(t)$. Supposed it is achieved at $t'$. Although $t'$ may not be where the maximum of $|I(t)| - c(\lambda(t))^2$ is, it tells us whether the maximum is positive. If not, we can simply terminate the process. If it is, we know that $l$ intersects $\mathcal{D}$, and $p + v t' \in \mathcal{D}$. This allows us to use divide-and-conquer as in \cite{avis1996}, since we can recurse in the direction of $t'$ whenever we query a point $t$ and find $f(t) < 0$.
\end{proof}

\begin{proposition}\label{prop_modified_plane}
    Let $S \subset \mathbb{R}^3$ be a plane. If $S$ is horizontal or if $S$ intersects the polygon $- P \subset \mathcal{D}$, then we can perform \textbf{PlaneDecision} on $S$ in $O(n \log n)$ time.
\end{proposition}
\begin{proof}
If $S$ is horizontal, then we can apply \Cref{pseudo_plane}. If the maximum is negative, then we simply report the side of $S$ containing $-P$, otherwise we proceed as in \Cref{lemma_side_decision}. 

Now assume $S$ is non-horizontal and intersects $-P$. Let $s = S \cap (-P)$. Then we know that $s \subset \mathcal{D}$. Let $l \subset S$ be a line we want to run the subroutine \textbf{LineDecision} on. By \Cref{lemma_modified_line_max}, we can find $\operatorname{max}_{\varphi \in l} f(\varphi)$ or report it is negative in $O(n)$ time. If it is the latter case, we report the side of $l$ containing $s$. Otherwise, $l$ intersects $\mathcal{D}$, and we can proceed as in \Cref{lemma_side_decision}. Thus, we can still find $\operatorname{max}_{\varphi \in S} f(\varphi)$ in $O(n \log n)$ time. Since $S$ intersects $\mathcal{D}$, we can use \Cref{lemma_side_decision} to complete \textbf{PlaneDecision} on $S$.
\end{proof}

\begin{theorem}\label{thm_general_symm_diff}
    Let $P$ and $Q$ be convex polygons with $n$ vertices in total. Suppose $\kappa \in (0, 1)$ is a constant. We can find a homothety $\varphi$ that minimizes 
    \[
        h(\varphi) = 2(1 - \kappa)|P \setminus \varphi(Q)| + 2\kappa |\varphi(Q) \setminus P|
    \]
    in $O(n \log^2 n)$ time. 
\end{theorem}
\begin{proof}
    We want to maximize $f(x, y, \lambda) = |C \cap (P + (-x, -y, \lambda))| - c \lambda^2$ over $\mathbb{R}^3$, where $c = \kappa |Q|$. 
    In order to apply our algorithm for \Cref{thm_algo}, we need to show that we only run \textbf{PlaneDecision} on horizontal planes and planes that intersect $- P$. 
    
    In the first stage (as outlined in \Cref{pseudo_algo}), we only run \textbf{PlaneDecision} on horizontal planes. 
    
    In the second stage, we apply \Cref{thm_prune_and_search} to the $O(n)$ groups of $O(n)$ lines that are the projections of the lines containing edges of event polygons on the $xz$-plane. Observe that these lines all intersect the projection of $- P$ on the $xz$-plane. In each recursive step of our algorithm, we query a horizontal (parallel to the $x$-axis) line and a line that goes ``between'' two lines in the $O(n^2)$ lines. The planes they represent both satisfy the condition for \Cref{prop_modified_plane}. Then we run \textbf{PlaneDecision} $O(1)$ more times to triangulate our feasible region. Here, we make a small modification: instead of maintaining a triangular feasible region, we maintain a trapezoidal one by making $O(1)$ horizontal cuts to make the region a trapezoid. 
    
    In the third stage, we have a ``tube'' and $O(n)$ event polygons that intersect it. As usual, we recursively construct a $(1/2)$-cutting by \Cref{lemma_cutting}. Chazelle's algorithm \cite{chazelle1993} picks $O(1)$ planes intersecting the target region as the cutting, along with $O(1)$ extra planes to triangulate each piece. All the planes containing the event polygons intersect $- P$, so we can run \textbf{PlaneDecision} on them. Instead of triangulating our target region, it suffices to reduce it to constant complexity. We do this by cutting it with $O(1)$ horizontal planes such that the remaining region only has vertices on two levels. Then, let $e$ be any non-horizontal edge. With $O(1)$ planes through $e$, we can cut the target region into prisms and pyramids with triangular bases. These planes all intersect $- P$ since they are between the two faces of the target region containing $e$, and the planes containing them intersect $- P$. 

    Therefore, with slight modifications to \Cref{thm_algo}, we obtain a deterministic $O(n \log^2 n)$ algorithm for minimizing $h(\varphi)$.
\end{proof}

\Cref{thm_symmetric_difference} follows as a direct corollary of \Cref{thm_general_symm_diff}.

\printbibliography
\end{document}